\documentclass[journal,onecolumn,draftclsnofoot,12pt,twoside]{IEEEtranTCOM}
\usepackage{epsfig,amsthm}

\usepackage{cite}

\usepackage[cmex10]{amsmath}

\usepackage{array}

\newtheorem{theorem}{Theorem}
\newcommand{\ie}{\textit{i.e.}}

\usepackage{ulem,color}
\newcommand{\newt}[1]{\textcolor{black}{#1}}

\begin{document}

\title{Adaptive Read Thresholds for NAND Flash}
%
%
%
\author{
    \IEEEauthorblockN{Borja Peleato\IEEEauthorrefmark{1}, Rajiv Agarwal\IEEEauthorrefmark{2}, John Cioffi\IEEEauthorrefmark{2}, Minghai Qin\IEEEauthorrefmark{3}, Paul H. Siegel\IEEEauthorrefmark{3}}\\
    \IEEEauthorblockA{\IEEEauthorrefmark{1}Purdue University}
    \IEEEauthorblockA{\IEEEauthorrefmark{2}Stanford University}
    \IEEEauthorblockA{\IEEEauthorrefmark{3}UCSD}

    \thanks{\copyright2015 IEEE.  Personal use of this material is permitted.  Permission from IEEE must be obtained for all other uses, in any current or future media, including reprinting/republishing this material for advertising or promotional purposes, creating new collective works, for resale or redistribution to servers or lists, or reuse of any copyrighted component of this work in other works. DOI: 10.1109/TCOMM.2015.2453413}
}

\maketitle
\vspace{-2cm}
\begin{abstract}

A primary source of increased read time on NAND flash comes from the fact that in the presence of noise, the flash medium must be read several times using different read threshold voltages for the decoder to succeed. This paper proposes an algorithm that uses a limited number of re-reads to characterize the noise distribution and recover the stored information. Both hard and soft decoding are considered. For hard decoding, the paper attempts to find a read threshold minimizing bit-error-rate (BER) and derives an expression for the resulting codeword-error-rate. For soft decoding, it shows that minimizing BER and minimizing codeword-error-rate are competing objectives in the presence of a limited number of allowed re-reads, and proposes a trade-off between the two.

The proposed method does not require any prior knowledge about the noise distribution, but can take advantage of such information when it is available. Each read threshold is chosen based on the results of previous reads, following an optimal policy derived through a dynamic programming backward recursion. The method and results are studied from the perspective of an SLC Flash memory with Gaussian noise for each level but the paper explains how the method could be extended to other scenarios.
\end{abstract}

%

\IEEEpeerreviewmaketitle

\section{Introduction}
\label{s-intro}

\subsection{Overview}

The introduction of Solid
State Drives (SSD) based on NAND flash memories has
revolutionized mobile, laptop,
and enterprise storage by offering random access to the information with dramatically higher read throughput and power-efficiency than hard disk drives. However, SSD's are considerably more expensive, which poses an obstacle to their widespread use. NAND flash manufacturers have tried to pack more data in the same silicon area by scaling the size of the flash cells and storing more bits in each of them, thus reducing the cost per gigabyte (GB) and making flash more attractive to consumers, but this cell-size shrinkage has come at the cost of reduced performance. As cell-size shrinks to sub-16nm limits, noise can cause the voltage residing on the cell at read time to be significantly different from the voltage that was intended to be stored at the time of write. Even in current state-of-the-art 19nm NAND, noise is significant towards the end of life of the drive. One way to recover host data in the presence of noise is to use advanced signal processing algorithms~\cite{peleato2012maximizing},~\cite{zhou2011error},~\cite{peleato2012probabilistic},~\cite{asadi2014optimal}, but excessive re-reads and post-read signal processing could jeopardize the advantages brought by this technology.

Typically, all post-read signal processing algorithms require re-reads using different thresholds, but the default read thresholds, which are good for voltage levels intended during write, are often suboptimal for read-back of host data. Furthermore, the noise in the stored voltages is random and depends on several factors such as time, data, and temperature; so a fixed set of read thresholds will not be optimal throughout the entire life of the drive. Thus, finding optimal read thresholds in a dynamic manner to minimize BER and speed up the post-processing is essential.

The first half of the paper proposes an algorithm for characterizing the distribution of the noise for each nominal voltage level and estimating the read thresholds which minimize BER. It also presents an analytical expression relating the BER found using the proposed methods to the minimum possible BER. Though BER is a useful metric for algebraic error correction codes, the distribution of the number of errors is also important. Some flash memory controllers use a weaker decoder when the number of errors is small and switch to a stronger one when the former fails, both for the same code (e.g. bit-flipping and min-sum for decoding an LDPC code  \newt{\cite{anholt2013dual}}).
The average read throughput and total power consumption depends on how frequently each decoder is used. Therefore, the distribution of the number of errors, which is also derived here, is a useful tool to find NAND power consumption.

The second half of the paper modifies the proposed algorithm to address the quality of the soft information generated, instead of just the number of errors. In some cases, the BER is too large for a hard decoder to succeed, even if the read is done at the optimal threshold. It is then necessary to generate soft information by performing multiple reads with different read thresholds. The choice of read thresholds has a direct impact on the quality of the soft information generated, which in turn dictates the number of decoder iterations and the number of re-reads required. The paper models the flash as a discrete memoryless channel with mismatched decoding and attempts to maximize its capacity through dynamic programming.

\newt{The overall scheme will work as follows. First, the controller will read with an initial threshold and attempt a hard-decoding of the information. If the noise is weak and the initial threshold was well chosen, this decoding will succeed and no further processing will be needed. Otherwise, when this first decoding fails, the controller will perform additional reads with adaptively chosen thresholds to estimate the mean and/or variance of the voltage values for each level. These estimates will in turn be used to estimate the minimum feasible BER and the corresponding optimal read threshold. The flash controller then decides whether to perform an additional read with that estimated threshold to attempt hard decoding again, or directly attempt a more robust decoding of the information, leveraging the reads already performed to generate soft information.} 

\subsection{Literature review}
\label{ss-literatureReview}

Most of the existing literature on optimizing the read thresholds for NAND flash assumes that prior information on the noise is available \newt{(e.g., \cite{dong2011use},\cite{sala2013dynamic},\cite{cai2012error},\cite{li2014noise},\cite{cai2013threshold})}. Some methods, such at the one proposed by Wang et al. in~\cite{wang2011soft}, assume complete knowledge of the noise and choose the read thresholds so as to maximize the mutual information between the values written and read, while others attempt to predict the noise from the number of program-erase (PE) cycles and then optimize the read thresholds based on that prediction. An example of the latter was proposed by Cai et al. in~\cite{cai2013program}. \newt{Other references addressing threshold selection and error-correction codes are \cite{gabrys2014coding} and \cite{gabrys2013graded}}.

However, in some practical cases there is no prior information available, or the prior information is not accurate enough to build a reliable noise model. In these situations, a common approach is to perform several reads with different thresholds searching for the one that returns an equal number of cells on either side, \ie, the median between the two levels\footnote{In many cases this threshold is not explicitly identified as the median cell voltage, but only implicitly as the solution of $\frac{t-\mu_1}{\sigma_1}=\frac{t-\mu_2}{\sigma_2}$, where $(\mu_1,\sigma_1)$ and $(\mu_2,\sigma_2)$ are the mean and standard deviation of the level voltages.}. However, the median threshold is suboptimal in general, as was shown in~\cite{peleato2012maximizing}. In~\cite{zhou2011error}~and~\cite{zhou2011nonuniform} Zhou et al. proposed encoding the data using balanced, asymmetric, or Berger codes to facilitate the threshold selection. Balanced codes guarantee that all codewords have the same number of ones and zeros, hence narrowing the gap between the median and optimal thresholds. Asymmetric and Berger codes, first described in~\cite{berger1961note}, leverage the known asymmetry of the channel to tolerate suboptimal thresholds. Berger codes are able to detect {\bf any} number of unidirectional errors. In cases of significant leakage, where all the cells reduce their voltage level, it is possible to perform several reads with progressively decreasing thresholds until the Berger code detects a low enough number of errors, and only then attempt decoding to recover the host information.

Researchers have also proposed some innovative data representation schemes with different requirements in terms of read thresholds. For example, rank modulation~\cite{jiang2008error},\cite{jiang2009rank},\newt{\cite{en2011compressed},\cite{li2012compressed},\cite{gad2013rank}} stores information in the relative voltages between the cells instead of using pre-defined voltage levels. \newt{The strategy of writing data represented by rank modulation in parallel to flash memories is studied in~\cite{qin2013parallel}.} Theoretically, rank modulation does not require actual read thresholds, but just comparisons between the cell voltages. Unfortunately, there are a few technological challenges that need to be overcome before rank modulation becomes practical. Other examples include constrained codes~\cite{qin2014constrained},\newt{\cite{kayser2014constructions}; write-once memories codes~\cite{gabrys2015constructions},\cite{yaakobi2012multiple},\cite{gabrys2015constructions},\cite{bhatia2014lattice}; and other rewriting codes~\cite{li2013polar}. All these codes} impose restrictions on the levels that can be used during a specific write operation. Since read thresholds need only separate the levels being used, they can often take advantage of these restrictions.

The scheme proposed in this paper is similar to those described in~\cite{papandreou2014using}~and~\cite{lee2013estimation} in that it assumes no prior information about the noise or data representation, but it is significantly simpler and more efficient.  We propose using a small number of reads chosen by a dynamic program to simultaneously estimate the noise and recover the information, instead of periodically testing multiple thresholds (as in~\cite{papandreou2014using}) or running a computationally intensive optimization algorithm to perfect the model (as in~\cite{lee2013estimation}). A prior version of this paper was published in~\cite{peleato2012towards}, but the work presented here has been significantly extended with a bound on the capacity for the soft decoding case and a dynamic programming method for optimizing the read thresholds.

\section{System model}
\label{s-sysmodel}

Cells in a NAND flash are organized in terms of pages, which are the smallest units for write and read operations. Writing the cells in a page is done through a program and verify approach where voltage pulses are sent into the cells until their stored voltage exceeds the desired one. Once a cell has reached its desired voltage, it is inhibited from receiving subsequent pulses and the programming of the other cells in the page continues. However, the inhibition mechanism is non-ideal and future pulses may increase the voltage of the cell~\cite{cai2013program}, creating write noise.
The other two main sources of noise are inter-cell interference (ICI), caused by interaction between neighboring cells~\cite{dong2010using}, and charges leaking out of the cells with time and heat~\cite{torsi2011program}.

Some attempts have been made to model these sources of noise as a function of time, voltage levels, amplitude of the programming pulses, etc. Unfortunately, the noise is temperature- and page-dependent as well as time- and data-dependent~\cite{yaakobi2010error}. Since the controller cannot measure those factors, it cannot accurately estimate the noise without performing additional reads. This paper assumes that the overall noise follows a Gaussian distribution for each level, as is common in the literature, but assumes no prior knowledge about their means or variances. Section~\ref{s-Extension} will explain how the same idea can be used when the noise is not Gaussian.

Reading the cells in a page is done by comparing their stored voltage with a threshold voltage $t$. The read operation returns a binary vector with one bit for each cell. Bits corresponding to cells with voltage lower than $t$ are 1 and those corresponding to cells with voltage higher than $t$ are 0. However, the aforementioned sources of voltage disturbance can cause some cells to be misclassified, introducing errors in the bit values read. The choice of a read threshold therefore becomes important to minimize the BER in the reads.

In a $b$-bit MLC flash, each cell stores one of $2^b$ distinct predefined voltage levels. When each cell stores multiple bits, i.e. $b\geq 2$, the mapping of information bits to voltage levels is done using Gray coding to ensure that only one bit changes between adjacent levels. Since errors almost always happen between adjacent levels, Gray coding minimizes the average BER. Furthermore, each of the $b$ bits is assigned to a different page, as shown in Fig.~\ref{f-grayMapping}. This is done so as to reduce the number of comparisons required to read a page. For example, the lower page of a TLC ($b=3$) flash can be read by comparing the cell voltages with a single read threshold located between the fourth and fifth levels, denoted by D in Fig.~\ref{f-grayMapping}. The first four levels encode a bit value 1 for the lower page, while the last four levels encode a value 0. Unfortunately, reading the middle and upper pages require comparing the cell voltages with more read thresholds: two \newt{(B,F)} for the middle page and four \newt{(A,C,E,G)} for the upper page.
\begin{figure}[htb]
\begin{minipage}[b]{1.0\linewidth}
\centering
\centerline{\epsfig{figure=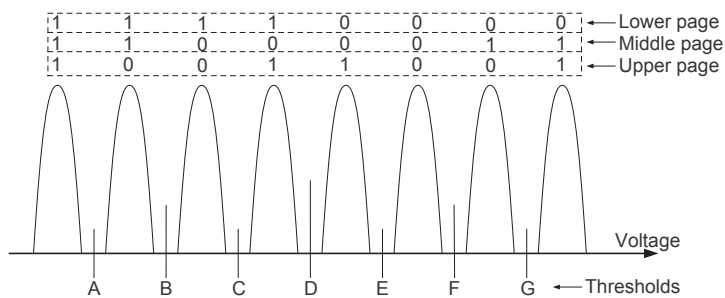,width=.6\linewidth}}
\end{minipage}
\caption{Typical mapping of bits to pages and voltage levels in a TLC Flash memory.}
\label{f-grayMapping}
\end{figure}

Upper pages take longer to read than lower ones, but the difference is not as large as it might seem. Flash chips generally incorporate dedicated hardware for performing all the comparisons required to read upper pages, without the additional overhead that would arise from issuing multiple independent read requests. The flash controller can then be oblivious to the type of page being read. Read commands only need to specify the page being read and a scalar parameter representing the desired shift in the read thresholds from their default value. If the page is a lower one, which employs only one threshold, the scalar parameter is understood as the amount by which this threshold needs to be shifted. If the page is an upper one, which employs multiple read thresholds, their shifts are parameterized by the scalar parameter. For example, a parameter value of $\Delta$ when reading the middle page in Fig.~\ref{f-grayMapping} could shift thresholds $B$ and $F$ by $\Delta$ and $-\frac{3}{2}\Delta$ mV, respectively. Then, cells whose voltage falls between the shifted thresholds $B$ and $F$ would be read as $0$ and the rest as $1$.

After fixing this parametrization, the flash controller views all the pages in an MLC or TLC memory as independent SLC pages with a single read shift parameter that needs to be optimized. In theory, each low level threshold could be independently optimized, but the large number of reads and amount of memory required would render that approach impractical. Hence, most of the paper will assume a SLC architecture for the flash and Section~\ref{s-Extension} will show how the same method and results can be readily extended to memories with more bits per cell.

Figure~\ref{f-thresholds} (a) shows two overlapping Gaussian probability density functions (pdfs), corresponding to the two voltage levels to which cells can be programmed. Since data is generally compressed before being written onto flash, approximately the same number of cells is programmed to each level. The figure also includes three possible read thresholds. Denoting by $(\mu_1,\sigma_1)$ and $(\mu_2,\sigma_2)$ the means and standard deviations of the two Gaussian distributions, the thresholds are: $t_\mathrm{mean}=\frac{\mu_1+\mu_2}{2}$, $t_\mathrm{median}=\frac{\mu_1\sigma_2+\mu_2\sigma_1}{\sigma_1+\sigma_2}$, and $t^\star$, which minimizes BER. If the noise variance was the same for both levels all three thresholds would be equal, but this is not the case in practice. The plot legend provides the BER obtained when reading with each of the three thresholds.
\begin{figure}[htb]
\centering
\begin{tabular}{cc}
  \epsfig{figure=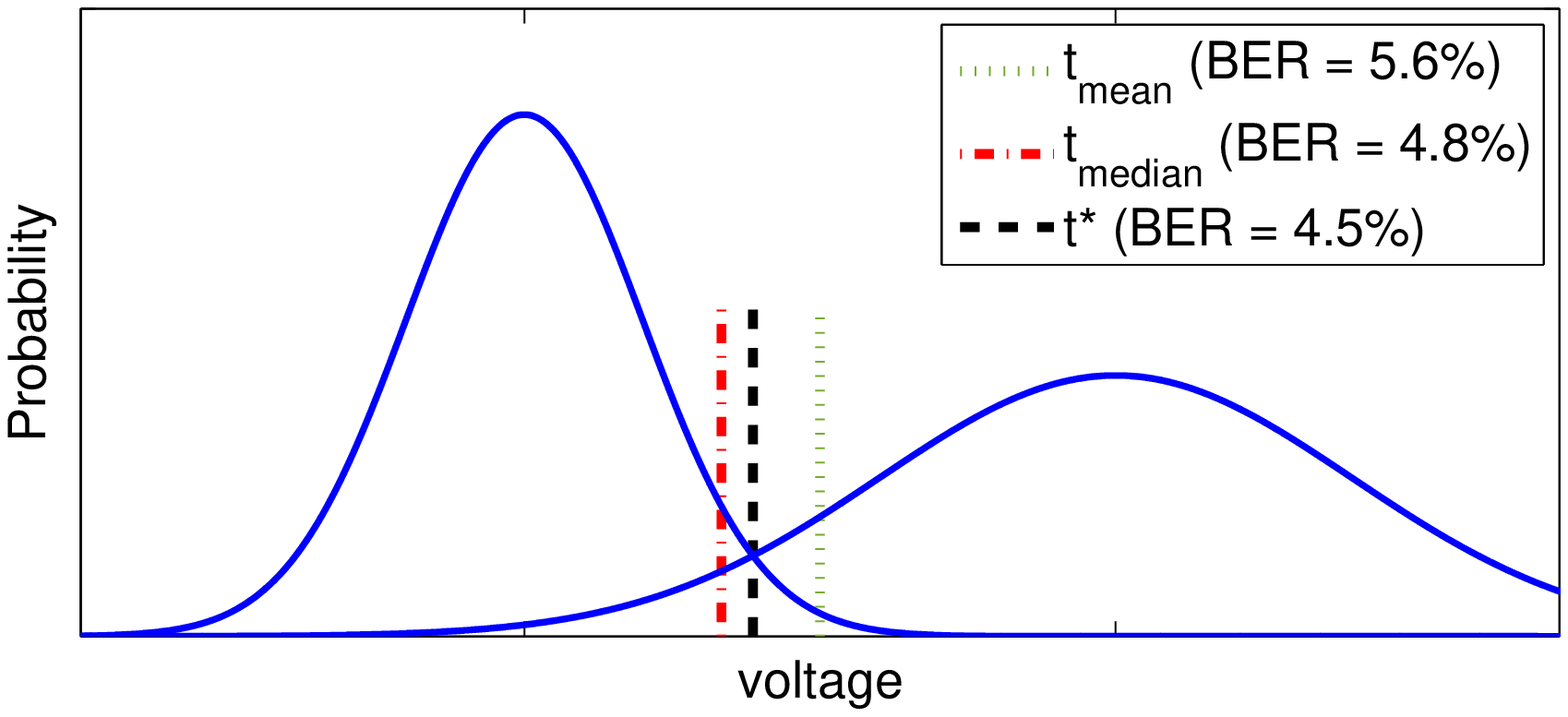,width=8cm} & \epsfig{figure=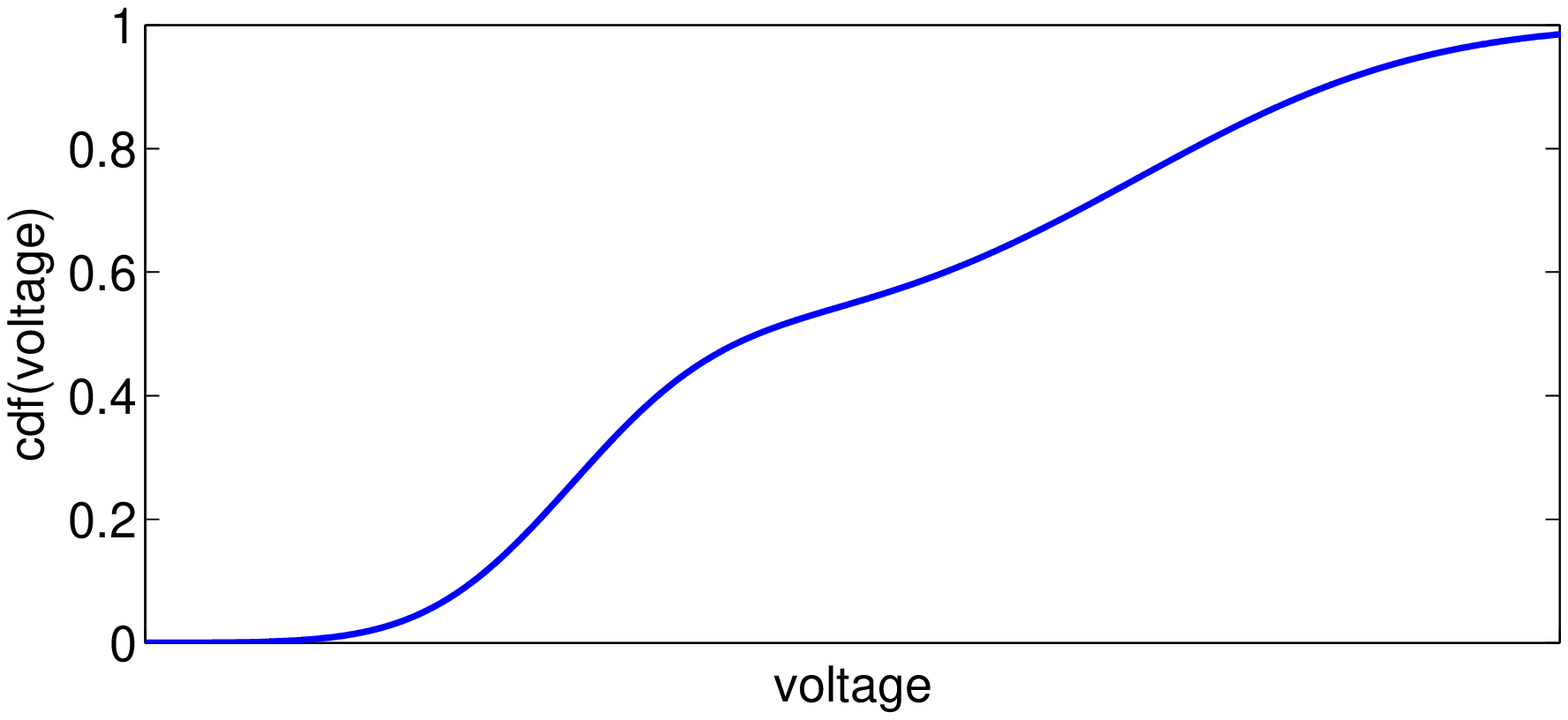,width=8cm}\\
  (a) & (b)
  \end{tabular}
\caption{(a): Cell voltages pdf in an SLC page, and BER for three different thresholds: $t_\mathrm{mean}=(\mu_1+\mu_2)/2$ is the average of the cell voltages, $t_\mathrm{median}$ returns the same number of 1s and 0s and $t^\star$ minimizes the BER and is located at the intersection of the two pdfs. (b): cdf corresponding to pdf in (a).}
\label{f-thresholds}
\end{figure}

There exist several ways in which the optimal threshold, $t^\star$, can be found. A common approach is to perform several reads by shifting the thresholds in one direction until the decoding succeeds. Once the data has been recovered, it can be compared with the read outputs to find the threshold \newt{yielding the lowest BER~\cite{papandreou2014using}}. However, this method \newt{can require a large number of reads if the initial estimate is inaccurate}, which reduces read throughput, and additional memory to store \newt{and compare} the successive reads, which increases cost. The approach taken in this paper consists of estimating $(\mu_1,\sigma_1)$ and $(\mu_2,\sigma_2)$ and deriving $t^\star$ analytically. It will be shown how this can be done with as few reads as possible, thereby reducing read time. Furthermore, the mean and standard deviation estimates can also be used for other tasks, such as generating soft information for LDPC decoding.

A read operation with a threshold voltage $t$ returns a binary vector with a one for each cell whose voltage level is lower than $t$ and zero otherwise. The fraction of ones in the read output is then equal to the probability of a randomly chosen cell having a voltage level below $t$. Consequently, a read with a threshold voltage $t$ can be used to obtain a sample from the cumulative distribution function (cdf) of cell voltages at $t$, illustrated in Fig.~\ref{f-thresholds} (b).

The problem is then reduced to estimating the means and variances of a mixture of Gaussians using samples from their joint cdf. These samples will be corrupted by model, read, and quantization noise. Model noise is caused by the deviation of the actual distribution of cell voltages from a Gaussian distribution. Read noise is caused by the intrinsic reading mechanism of the flash, which can read some cells as storing higher or lower voltages than they actually have. Quantization noise is caused by limited computational accuracy and rounding of the Gaussian cdf\footnote{Since the Gaussian cdf has no analytical expression, it is generally quantized and stored as a lookup table}. All these sources of noise are collectively referred to as read noise in this paper. It is assumed to be zero mean, but no other restriction is imposed in our derivations.

It is desirable to devote as few reads as possible to the estimation of $(\mu_1,\sigma_1)$ and $(\mu_2,\sigma_2)$. The accuracy of the estimates would improve with the number of reads, but read time would also increase. Since there are four parameters to be estimated, at least four reads will be necessary. Section~\ref{s-minBER} describes how the locations of the read thresholds should be chosen in order to achieve accurate estimates and Section~\ref{s-ldpc} extends the framework to consider how these reads could be reused to obtain soft information for an LDPC decoder. If the soft information obtained from the first four reads is enough for the LDPC decoding to succeed, no additional reads will be required, thereby reducing the total read time of the flash. Section~\ref{s-dynamicProgramming} proposes a dynamic programming method for optimizing the thresholds for a desired objective. Finally, Section~\ref{s-Extension} explains how to extend the algorithm for MLC or TLC memories, as well as for non-Gaussian noise distributions. Section~\ref{s-NumResults} provides simulation results to evaluate the performance of the proposed algorithms and Section~\ref{s-conclusion} concludes the paper.

\section{Hard decoding: minimizing BER}
\label{s-minBER}

\subsection{Parameter estimation}
\label{ss-paramEstimation}
Let $t_i$, $i=1,\ldots,4$ be four voltage thresholds used for reading a page and let $y_i$, $i=1,\ldots,4$ be the fraction of ones in the output vector for each of the reads, respectively. If $(\mu_1,\sigma_1)$ and $(\mu_2,\sigma_2)$ denote the voltage mean and variance for the cells programmed to the two levels, then
\begin{equation}\label{e-mixGaussCdf}
y_i=\frac{1}{2}Q\left(\frac{\mu_1-t_i}{\sigma_1}\right)+\frac{1}{2}Q\left(\frac{\mu_2-t_i}{\sigma_2}\right)+n_{y_i}, \qquad i=1,\ldots,4,
\end{equation}
where
\begin{equation}\label{e-qfunct}
Q(x) = \int_{x}^\infty(2\pi)^{-\frac{1}{2}}e^{-\frac{t^2}{2}}dt
\end{equation}
and $n_{y_i}$ denotes the read noise associated to $y_i$. In theory, it is possible to estimate $(\mu_1,\sigma_1)$ and $(\mu_2,\sigma_2)$ from  $(t_i,y_i)$, $i=1,\ldots,4$ by solving the system of non-linear equations in Eq.~(\ref{e-mixGaussCdf}), but in practice the computational complexity \newt{could be too large for some systems}. Another possible approach would be to restrict the estimates to a pre-defined set of values and generate a lookup table for each combination. Finding the table which best fits the samples would require negligible time but the amount of memory required \newt{could render this approach impractical for some systems}. This section proposes and evaluates a progressive read algorithm that combines these two approaches, providing similar accuracy to the former and requiring only a standard normal \newt{($\mu=0$, $\sigma=1$)} look-up table.

{\bf\emph{Progressive Read Algorithm:}} The key idea is to perform two reads at locations where one of the $Q$ functions is known to be either close to 0 or close to 1. The problem with solving the system in Eq.~(\ref{e-mixGaussCdf}) was that a sum of $Q$ functions cannot be easily inverted. However, once one of the two $Q$ functions is fixed at 0 or 1, the equation can be put in linear form using a standard normal table to invert the other $Q$ function. The system of linear equations can then be solved to estimate the first mean and variance. Once the first mean and variance have been estimated they can be used to evaluate a $Q$ function from each of the two remaining equations in Eq.~(\ref{e-mixGaussCdf}), which can then be solved in a similar way. For example, if $t_1$ and $t_2$ are significantly smaller than $\mu_2$, then
\[
Q\left(\frac{\mu_2-t_1}{\sigma_2}\right)\simeq 0 \simeq Q\left(\frac{\mu_2-t_2}{\sigma_2}\right)
\]
 and Eq.~(\ref{e-mixGaussCdf}) can be solved for $\widehat{\mu_1}$ and $\widehat{\sigma_1}$ to get
 \begin{equation}
   \widehat{\sigma_1} = \frac{t_2-t_1}{Q^{-1}(2y_1)-Q^{-1}(2y_2)} \qquad \widehat{\mu_1} = t_2 + \widehat{\sigma_1} Q^{-1}(2y_2)\label{e-estimators}.
   \end{equation}
Substituting these in the equations for the third and fourth reads and solving gives
   \begin{equation}
   \widehat{\sigma_2} = \frac{t_4-t_3}{Q^{-1}(2y_3-q_3)-Q^{-1}(2y_4-q_4)} \qquad \widehat{\mu_2} = t_4 + \widehat{\sigma_2} Q^{-1}(2y_4-q_4),\label{e-estimators2-b}
   \end{equation}
where
   \[
   q_3=Q\left(\frac{\widehat{\mu_1}-t_3}{\widehat{\sigma_1}}\right) \qquad q_4=Q\left(\frac{\widehat{\mu_1}-t_4}{\widehat{\sigma_1}}\right).
   \]

   It could be argued that, since the pdfs are not known a priori, it is not possible to determine two read locations where one of the $Q$ functions is close to 0 or close to 1. In practice, however, each read threshold can be chosen based on the result from the previous ones. For example, say the first randomly chosen read location returned $y_1=0.6$. This read, if used for estimating the higher level distribution, will be a bad choice because there will be significant overlap from the lower level. Hence, a smart choice would be to obtain two reads for the lower level that are clear of the higher level by reading to the far left of $t_1$. Once the lower level is canceled, the $y_1=0.6$ read can be used in combination with a fourth read to the right of $t_1$ to estimate the higher level distribution.

   Once the mean and variance of both pdfs have been estimated, it is possible to derive an estimate for the read threshold minimizing the BER. The BER associated to a given read threshold $t$ is given by
\begin{equation}\label{e-t2BER}
\mathrm{BER}(t)=\frac{1}{2}\left(Q\left(\frac{\mu_2-t}{\sigma_2}\right)+1-Q\left(\frac{\mu_1-t}{\sigma_1}\right)\right).
\end{equation}
Making its derivative equal to zero gives the following equation for the optimal threshold $t^\star$
\begin{equation}\label{e-intersection}
\frac{1}{\sigma_2}\phi\left(\frac{\mu_2-t^\star}{\sigma_2}\right)=\frac{1}{\sigma_1}\phi\left(\frac{\mu_1-t^\star}{\sigma_1}\right),
\end{equation}
where $\phi(x)=(2\pi)^{-(1/2)}e^{-x^2/2}$. The optimal threshold $t^\star$ is located at the point where both Gaussian pdfs intersect. An estimate $\widehat{t^\star}$ for $t^\star$ can be found from the following quadratic equation
\begin{equation}\label{e-threshold}
2\log\left(\frac{\widehat{\sigma_2}}{\widehat{\sigma_1}}\right)=\left(\frac{\widehat{t^\star}-\widehat{\mu_1}}{\widehat{\sigma_1}}\right)^2-\left(\frac{\widehat{t^\star}-\widehat{\mu_2}}{\widehat{\sigma_2}}\right)^2,
\end{equation}
which can be shown to be equivalent to solving Eq.~(\ref{e-intersection}) with $(\mu_1,\sigma_1)$ and $(\mu_2,\sigma_2)$ replaced by their estimated values.

If some parameters are known, the number of reads can be reduced. For example, if $\mu_1$ is known, the first read can be replaced by $t_1=\mu_1$, $y_1=0.25$ in the above equations. Similarly, if $\sigma_1$ is known $(t_1,y_1)$ are not required in Eqs.~(\ref{e-estimators})-(\ref{e-estimators2-b}).

\subsection{Error propagation}
\label{ss-errorprop}
This subsection first studies how the choice of read locations affects the accuracy of the estimators $(\widehat{\mu_1},\widehat{\sigma_1})$, $(\widehat{\mu_2},\widehat{\sigma_2})$, and correspondingly $\widehat{t^\star}$. Then it analyzes how the accuracy of $\widehat{t^\star}$ translates into BER$(\widehat{t^\star})$, and provides some guidelines as to how the read locations should be chosen. Without loss of generality, it will be assumed that $(\mu_1,\sigma_1)$ are estimated first using $(t_1,y_1)$ and $(t_2,y_2)$ according to the Progressive Read Algorithm described in Section~\ref{ss-paramEstimation}, and $(\mu_2,\sigma_2)$ are estimated in the second stage. In this case, Eq.~(\ref{e-mixGaussCdf}) reduces to
\[
Q\left(\frac{\mu_1-t_i}{\sigma_1}\right)=2y_i-2n_{y_i}
\]
for $i=1,2$ and the estimates are given by Eqs.~(\ref{e-estimators}).

If the read thresholds are on the tails of the distributions, a small perturbation in the cdf value $y$ could cause a significant change in $Q^{-1}(y)$. This will in turn lead to a significant change in the estimates. Specifically, a first-order Taylor expansion of $Q^{-1}(y+n_y)$ at $y$ can be written as
\begin{equation}\label{e-errorQinv}
Q^{-1}(y+n_y) = x - \sqrt{2\pi}e^{\frac{x^2}{2}}n_y + O(n_y^2),
\end{equation}
where $x=Q^{-1}(y)$. Since the exponent of $e$ is always positive, the first-order error term is minimized when $x=0$, i.e., when the read is performed at the mean. The expressions for $(\widehat{\mu_1},\widehat{\sigma_1})$ and $(\widehat{\mu_2},\widehat{\sigma_2})$ as seen in Eqs.~(\ref{e-estimators})-(\ref{e-estimators2-b}) use inverse $Q$ functions, so the estimation error due to read noise will be reduced when the reads are done close to the mean of the Gaussian distributions. The first order Taylor expansion of Eq.~(\ref{e-estimators}) at $\sigma_1$ is given by
\begin{align}
\widehat{\sigma_1}  &= \sigma_1 - \frac{\sigma_1^2}{t_2-t_1}(n_2-n_1) + O(n_1^2,n_2^2)\label{e-errorvar}
\end{align}
where
\begin{equation}
n_1 = 2\sqrt{2\pi}e^{\frac{(t_1-\mu_1)^2}{2\sigma_1^2}}n_{y_1} + O(n_{y_1}^2) \qquad n_2 = 2\sqrt{2\pi}e^{\frac{(t_2-\mu_1)^2}{2\sigma_1^2}}n_{y_2} + O(n_{y_2}^2).\label{e-errorQinv2}
\end{equation}
A similar expansion can be performed for $\widehat{\mu_1}$, obtaining
\begin{align}
\widehat{\mu_1}  & = \mu_1 - \sigma_1\frac{(t_2-\mu_1)n_1 - (t_1-\mu_1)n_2}{t_2-t_1} + O(n_1^2,n_2^2).\label{e-errormean}
\end{align}

Two different tendencies can be observed in the above expressions. On one hand, Eqs.~(\ref{e-errorQinv2}) suggest that both $t_1$ and $t_2$ should be chosen close to $\mu_1$ so as to reduce the magnitude of $n_1$ and $n_2$. On the other hand, if $t_1$ and $t_2$ are very close together, the denominators in Eq.~(\ref{e-errorvar}) and (\ref{e-errormean}) can become small, increasing the estimation error.

 The error expansions for $\widehat{\mu_2}$, $\widehat{\sigma_2}$ and $\widehat{t^\star}$, are omitted for simplicity, but it can be shown that the dominant terms are linear in $n_{y_i}$, $i=1,\ldots,4$ as long as all $n_{y_i}$ are small enough. The Taylor expansion for BER$(\widehat{t^\star})$ at $t^\star$ is
\begin{align}
\mathrm{BER}(\widehat{t^\star}) & = \mathrm{BER}(t^\star) + \left(\frac{1}{2\sigma_2}\phi\left(\frac{\mu_2-t^\star}{\sigma_2}\right)-\frac{1}{2\sigma_1}\phi\left(\frac{\mu_1-t^\star}{\sigma_1}\right)\right)e_{t^\star} + O(e_{t^\star}^2)\nonumber\\
& =  \mathrm{BER}(t^\star) + O(e_{t^\star}^2),
\end{align}
where $\widehat{t^\star}=t^\star+e_{t^\star}$. The cancellation of the first-order term is justified by Eq.~(\ref{e-intersection}). Summarizing, the mean and variance estimation error increases linearly with the read noise, as does the deviation in the estimated optimal read threshold. The increase in BER, on the other hand, is free from linear terms. As long as the read noise is not too large, the resulting BER($\widehat{t^\star}$) is close to the minimum possible BER. The numerical simulations in Fig.~\ref{f-errorprop} confirm these results.
\begin{figure}[htb]
\begin{minipage}[b]{1.0\linewidth}
  \centering
  \centerline{\epsfig{figure=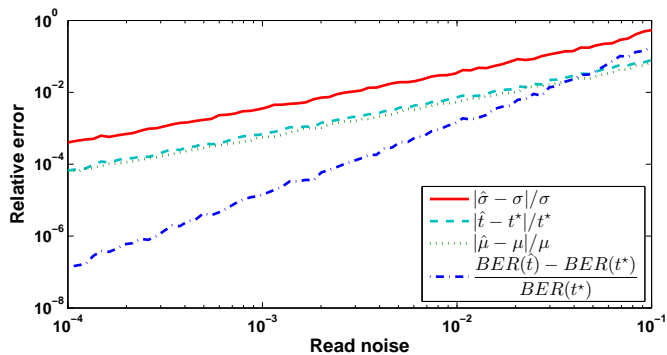,width=10cm}}
\end{minipage}
\caption{The relative error in the mean, variance, and threshold estimates increases linearly with the read noise (slope=1), but the relative increase in BER grows quadratically (slope=2) and is negligible for a wide range of read noise amplitudes.
}
\label{f-errorprop}
\end{figure}

In view of these results, it seems that the read thresholds should be spread out over both pdfs but close to the levels' mean voltages. Choosing the thresholds in this way will reduce the error propagating from the reads to the estimates. However, read thresholds can be chosen sequentially, using the information obtained from each read in selecting subsequent thresholds. Section~\ref{s-dynamicProgramming} proposes a method for finding the optimal read thresholds more precisely.

\section{Soft Decoding: Tradeoff BER-LLR}
\label{s-ldpc}
This section considers a new scenario where a layered decoding approach is used for increased error-correction capability. After reading a page, the controller may first attempt to correct any bit errors in the read-back codeword using a hard decoder alone, typically a bit-flipping hard-LDPC decoder~\cite{nguyen2011two}. Reading with the threshold $\widehat{t^\star}$ found through Eq.~(\ref{e-threshold}) reduces the number of hard errors but there are cases in which even BER($\widehat{t^\star}$) is too high for the hard decoder to succeed. When this happens, the controller will attempt a soft decoding, typically using a min-sum or sum-product soft LDPC decoder.

 Soft decoders are more powerful, but also significantly slower and less power efficient than hard decoders. Consequently, \newt{invoking soft LDPC decoding too often can significantly impact} the controller's average read time. In order to estimate \newt{the probability of requiring soft decoding}, one must look at the distribution of the number of errors, and not at BER alone. \newt{For example, if the number of errors per codeword is uniformly distributed between 40 and 60 and the hard decoder can correct 75 errors, soft decoding will never be needed. However, if the number of errors is uniformly distributed between 0 and 100 (same BER), soft decoding will be required to decode 25\% of the reads.} Section~\ref{ss-usage} addresses this topic.

 The error-correction capability of a soft decoder depends heavily on the quality of the soft information at its input. It is always possible to increase such quality by performing additional reads, but this decreases read throughput. Section~\ref{ss-llrReadLoc} shows how the Progressive Read Algorithm from the previous section can be modified to provide high quality soft information.

\subsection{Distribution of the number of errors}
\label{ss-usage}
Let $N$ be the number of bits in a codeword. Assuming that both levels are equally likely, the probability of error for any given bit, denoted $p_e$, is given in Eq.~(\ref{e-t2BER}). Errors can be considered independent, hence the number of them in a codeword follows a binomial distribution with parameters $N$ and $p_e$. Since $N$ is usually large, it becomes convenient to approximate the binomial by a Gaussian distribution with mean $Np_e$ and variance $Np_e(1-p_e)$, or by a Poisson distribution with parameter $Np_e$ when $Np_e$ is small.

Under the Gaussian approximation paradigm, a codeword fails to decode with probability $Q\left(\frac{\alpha-Np_e}{\sqrt{Np_e(1-p_e)}}\right)$, where $\alpha$ denotes the number of bit errors that can be corrected.
\begin{table}
\begin{center}
\begin{tabular}{|c|c c c|}
\hline
Failure rate & $p_e=0.008$ & $p_e=0.01$ & $p_e=0.012$\\
\hline
$\alpha = 23$ & 0.05 & 0.28 & 0.62 \\
$\alpha = 25$ & 0.016 & 0.15 & 0.46 \\
$\alpha = 27$ & 0.004& 0.07 & 0.31\\
\hline
\end{tabular}
\end{center}
\caption{Failure rate for a $N=2048$ BCH code \newt{as a function of probability of bit error $p_e$ and number of correctable bit errors $\alpha$}. }\label{t-errDistrib}\vspace{-1.5cm}
\end{table}
Table~\ref{t-errDistrib} shows that a small change in the value of $\alpha$ may increase significantly the frequency with which a stronger decoder is needed. This has a direct impact on average power consumption of the controller. The distribution of bit errors can thus be used to judiciously obtain a value of $\alpha$ in order to meet a power constraint.

\subsection{Obtaining soft inputs}
\label{ss-llrReadLoc}
After performing $M$ reads on a page, each cell can be classified as falling into one of the $M+1$ intervals between the read thresholds. The problem of reliably storing information on the flash is therefore equivalent to the problem of reliable transmission over a discrete memoryless channel (DMC), such as the one in Fig.~\ref{f-dmc_channel}. Channel inputs represent the levels to which the cells are written, outputs represent read intervals, and channel transition probabilities specify how likely it is for cells programmed to a specific level to be found in each interval at read time.

\begin{figure}[htb]
\begin{center}
\includegraphics[width=.25\linewidth]{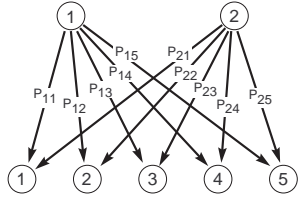}
\caption{DMC channel equivalent to Flash read channel with four reads.}
\label{f-dmc_channel}
\end{center}
\end{figure}

It is well known that the capacity of a channel is given by the maximum mutual information between the input and output over all input distributions (codebooks)~\cite{cover2012elements}. In practice, however, the code must be chosen at write time when the channel is still unknown, making it impossible to adapt the input distribution to the channel. Although some asymmetric codes have been proposed (e.g.~\cite{zhou2011nonuniform}\newt{,~\cite{kayser2014constructions}},~\cite{berman2011constrained}), channel inputs are equiprobable for most practical codes. The mutual information between the input and the output is then given by
\begin{align}
I(X;Y) 
& = \frac{1}{2}\sum_{j=1}^{M+1} p_{1j}\log(p_{1j})+p_{2j}\log(p_{2j})-(p_{1j}+p_{2j})\log\left(\frac{p_{1j}+p_{2j}}{2}\right), \label{e-symmCapacity}
\end{align}
where $p_{ij}$, $i=1,2$, $j=1,\ldots,M+1$ are the channel transition probabilities. For Gaussian noise, these transition probabilities can be found as
\begin{equation}\label{e-transProb}
p_{ij}=Q\left(\frac{\mu_i - t_{j}}{\sigma_i}\right) - Q\left(\frac{\mu_i - t_{j-1}}{\sigma_i}\right),
\end{equation}
where $t_0 = -\infty$ and $t_{M+1}=\infty$.

The inputs to a soft decoder are given in the form of log-likelihood ratios (LLR). The LLR value associated with a read interval $k$ is defined as $LLR_k=\log(p_{1k}/p_{2k})$. When the mean and variance are known it is possible to obtain good LLR values by reading at the locations that maximize $I(X;Y)$~\cite{wang2011soft}, which tend to be on the so-called uncertainty region, where both pdfs are comparable. However, the mean and variance are generally not known and need to be estimated. Section~\ref{s-minBER} provided some guidelines on how to choose read thresholds in order to obtain accurate estimates, but those reads tend to produce poor LLR values. Hence, there are two opposing trends: spreading out the reads over a wide range of voltage values yields more accurate mean and variance estimates but degrades the performance of the soft decoder, while concentrating the reads on the uncertainty region provides better LLR values but might yield inaccurate estimates which in turn undermine the soft decoding.

Some flash manufacturers are already incorporating soft read commands that return 3 or 4 bits of information for each cell, but the thresholds for those reads are often pre-specified and kept constant throughout the lifetime of the device. Furthermore, most controller manufacturers use a pre-defined mapping of read intervals to LLR values regardless of the result of the reads. We propose adjusting the read thresholds and LLR values adaptively to fit our channel estimates.

Our goal is to find the read locations that maximize the probability of successful decoding when levels are equiprobable and the decoding is done based on the estimated transition probabilities. With this goal in mind, Section~\ref{ss-bound} will derive a bound for the (symmetric and mismatched) channel capacity in this scenario and Section~\ref{s-dynamicProgramming} will show how to choose the read thresholds so as to maximize this bound. The \newt{error-free} coding rate specified by the bound will not be achievable in practice due to finite code length, limited computational power, etc., but the BER at the output of a decoder is closely related to the capacity of the channel~\cite{polyanskiy2010channel},~\cite{scarlett2013mismatched}. The read thresholds that maximize the capacity of the channel are generally the same ones that minimize the BER, \newt{in practice}.

\subsection{Bound for maximum transmission rate}
\label{ss-bound}

Shannon's channel coding theorem states that all transmission rates below the channel capacity are achievable when the channel is perfectly known to the decoder; unfortunately this is not the case in practice. The channel transition probabilities can be estimated by substituting the noise means and variances $\widehat{\mu_1},\widehat{\mu_2},\widehat{\sigma_1},\widehat{\sigma_2}$ 
into Eq.~(\ref{e-transProb}) but these estimates, denoted $\widehat{p_{ij}}$, $i=1,2$, $j=1,\ldots,5$, are inaccurate. The decoder is therefore not perfectly matched to the channel.

The subject of mismatched decoding has been of interest since the 1970's. The most notable early works are by Hui~\cite{hui1983fundamental} and Csisz\'{a}r and K\"{o}rner~\cite{csiszar1981graph}, who provided bounds on the maximum transmission rates under several different conditions. Merhav et al.~\cite{merhav1994information} related those results to the concept of generalized mutual information and, more recently, Scarlett et al.~\cite{scarlett2013mismatched} found bounds and error exponents for the finite code length case. It is beyond the scope of this paper to perform a detailed analysis of the mismatched capacity of a DMC channel with symmetric inputs; the interested reader can refer to the above references as well as~\cite{lapidoth1998reliable},~\cite{balakirsky1995converse},~\cite{alsan2014polarization}, and~\cite{arikan2009channel}. Instead, we will derive a simplified lower bound for the capacity of this channel in the same scenario that has been considered throughout the paper.

\begin{theorem} The maximum achievable rate of transmission with vanishing probability of error over a Discrete Memoryless Channel with equiprobable binary inputs, output alphabet $\mathcal{Y}$, transition probabilities $p_{ij}$, $i=1,2$, $j=1,\ldots,|\mathcal{Y}|$, and maximum likelihood decoding according to a different set of transition probabilities $\widehat{p_{ij}}$, $i=1,2$, $j=1,\ldots,|\mathcal{Y}|$ is lower bounded by
\begin{equation}\label{e-capacityHat}
C_{P,\hat{P}} = \frac{1}{2}\sum_{j=1}^{|\mathcal{Y}|} p_{1j}\log(\widehat{p_{1j}})+p_{2j}\log(\widehat{p_{2j}})-(p_{1j}+p_{2j})\log\left(\frac{\widehat{p_{1j}}+\widehat{p_{2j}}}{2}\right)
\end{equation}
\end{theorem}
\begin{proof}
Provided in the Appendix.
\end{proof}

It is worth noting that $C_{P,\hat{P}}$ is equal to the mutual information given in Eq.~(\ref{e-symmCapacity}) when the estimates are exact, and decreases as the estimates become less accurate. In fact, the probability of reading a given value $y\in\mathcal{Y}$ can be measured directly as the fraction of cells mapped to the corresponding interval, so it is usually the case that $\widehat{p_{1k}}+\widehat{p_{2k}}=p_{1k}+p_{2k}$. The bound then becomes $C_{P,\hat{P}}=I(X;Y)-D(P||\hat{P})$, where $I(X;Y)$ is the symmetric capacity of the channel with matched ML decoding and $D(P||\hat{P})$ is the relative entropy (also known as Kullback-Leibler distance) between the exact and the estimated transition probabilities.
\begin{align}
D(P||\hat{P}) & = \frac{1}{2}\sum_{j=1}^{|\mathcal{Y}|} p_{1j}\log\left(\frac{p_{1j}}{\widehat{p_{1j}}}\right)+p_{2j}\log\left(\frac{p_{2j}}{\widehat{p_{2j}}}\right).
\end{align}
In this case $C_{P,\hat{P}}$ is a concave function of the transition probabilities $(p_{ij},\widehat{p_{ij}})$, $i=1,2$, $j=1,\ldots,|\mathcal{Y}|$, since the relative entropy is convex and the mutual information is concave~\cite{cover2012elements}. The bound attains its maximum when the decoder is matched to the channel (i.e. $p_{ij}=\widehat{p_{ij}}\quad\forall i,j$) and the read thresholds are chosen so as to maximize the mutual information between $X$ and $Y$, but that solution is not feasible for our problem.

In practice, both the capacity of the underlying channel and the accuracy of the estimates at the decoder depend on the location of the read thresholds and cannot be maximized simultaneously.
%
Finding the read thresholds $t_1$, $t_2$, $t_3$, and $t_4$ which maximize $C_{P,\hat{P}}$ is not straightforward, but it can be done numerically. Section~\ref{s-dynamicProgramming} describes a dynamic programming algorithm for choosing each read threshold based on prior information about the noise and the result of previous reads.

\section{Optimizing read thresholds}
\label{s-dynamicProgramming}

In most practical cases, the flash controller has prior information about the voltage distributions, based on the number of PE cycles that the page has endured, its position within the block, etc. This prior information is generally not enough to produce accurate noise estimates, but it can be used to improve the choice of read thresholds. We wish to determine a policy to choose the optimal read thresholds sequentially, given the prior information about the voltage distributions and the results in previous reads.

This section proposes a dynamic programming framework to find the read thresholds that maximize the expected value of a user-defined reward function. If the goal is to minimize the BER at the estimated threshold $\widehat{t^\star}$, as in Section~\ref{s-minBER}, an appropriate reward would be $1-BER(\widehat{t^\star})$. If the goal is to maximize the channel capacity, the reward could be chosen to be $I(X;Y)-D(P\|\hat{P})$, as shown in Section~\ref{ss-bound}.

Let $\mathbf{x}=(\mu_1,\mu_2,\sigma_1,\sigma_2)$ and $\mathbf{r_i}=(t_i,y_i)$, $i=1,\ldots,4$ be vector random variables, so as to simplify the notation. If the read noise distribution $f_{n}$ is known, the prior distribution for $\mathbf{x}$ can be updated based on the result of each read $\mathbf{r_i}$ using Bayes rule and Eq.~(\ref{e-mixGaussCdf}):
\begin{align}
f_{\mathbf{x}|\mathbf{r_1},\ldots,\mathbf{r_i}} & = K\cdot f_{y_i|\mathbf{x},t_i}\cdot f_{\mathbf{x}|\mathbf{r_1},\ldots,\mathbf{r_{i-1}}}\nonumber\\
& = K\cdot f_n\left(y_i-\frac{1}{2}\left(Q\left(\frac{\mu_1-t_i}{\sigma_1}\right)+Q\left(\frac{\mu_2-t_i}{\sigma_2}\right)\right)\right)\cdot f_{\mathbf{x}|\mathbf{r_1},\ldots,\mathbf{r_{i-1}}},\label{e-distribUpdate}
\end{align}
where $K$ is a normalization constant. Furthermore, let $R(\mathbf{r_1},\mathbf{r_2},\mathbf{r_3},\mathbf{r_4})$ denote the expected reward associated with the reads $\mathbf{r_1},\ldots,\mathbf{r_4}$, after updating the prior $f_\mathbf{x}$ accordingly. In the following, we will use $R$ to denote this function, omitting the arguments for the sake of simplicity.

Choosing the fourth read threshold $t_4$ after the first three reads $\mathbf{r_1}$,\ldots,$\mathbf{r_3}$ is relatively straightforward: $t_4$ should be chosen so as to maximize the expected reward, given the results of the previous three reads. Formally,
\begin{align}
t_4^\star & = \arg\max_{t_4}E\left\{R|\mathbf{r_1},\ldots,\mathbf{r_3},t_4\right\},
\end{align}
where the expectation is taken with respect to $(y_4,\mathbf{x})$ by factoring their joint distribution in a similar way to Eq.~(\ref{e-distribUpdate}): $f_{y_4,\mathbf{x}|\mathbf{r_1},\ldots,\mathbf{r_3}}=f_{y_4|\mathbf{x},t_4}\cdot f_{\mathbf{x}|\mathbf{r_1},\ldots,\mathbf{r_3}}$.

This defines a policy $\pi$ for the fourth read, and a value $V_3$ for each possible state after the first three reads:
\begin{align}
\pi_4(\mathbf{r_1},\ldots,\mathbf{r_3}) & = t_4^\star\label{e-policy4}\\
V_3(\mathbf{r_1},\ldots,\mathbf{r_3}) & = E\left\{R|\mathbf{r_1},\ldots,\mathbf{r_3},t_4^\star\right\}.\label{e-value4}
\end{align}
In practice, the read thresholds $t_i$ and samples $y_i$ can only take a finite number of values, hence the number of feasible arguments in these functions (states) is also finite. This number can be fairly large, but it is only necessary to find the value for a small number of them, those which have non-negligible probability according to the prior $f_{\mathbf{x}}$ and value significantly larger than 0. For example, states are invariant to permutation of the reads so they can always be reordered such that $t_1<t_2<t_3$. Then, states which do not fulfill $y_1<y_2<y_3$ can be ignored. If the number of states after discarding meaningless ones is still too large, it is also possible to use approximations for the policy and value functions~\cite{bertsekas1995dynamic},~\cite{wang2014approximate}.

 Equations~(\ref{e-policy4})~and~(\ref{e-value4}) assign a value and a fourth read threshold to each meaningful state after three reads. The same idea, using a backward recursion, can be used to decide the third read threshold and assign a value to each state after two reads:
 \begin{align}
\pi_3(\mathbf{r_1},\mathbf{r_2}) & = \arg\max_{t_3}E\left\{V_3(\mathbf{r_1},\ldots,\mathbf{r_3})|\mathbf{r_1},\mathbf{r_2},t_3\right\}\label{e-policy3}\\
V_2(\mathbf{r_1},\mathbf{r_2}) & = \max_{t_3}E\left\{V_3(\mathbf{r_1},\ldots,\mathbf{r_3})|\mathbf{r_1},\mathbf{r_2},t_3\right\}\label{e-value3},
\end{align}
where the expectation is taken with respect to $(y_3,\mathbf{x})$. Similarly, for the second read threshold
\begin{align}
\pi_2(\mathbf{r_1}) & = \arg\max_{t_2}E\left\{V_2(\mathbf{r_1},\mathbf{r_2})|\mathbf{r_1},t_2\right\}\label{e-policy2}\\
V_1(\mathbf{r_1}) & = \max_{t_2}E\left\{V_2(\mathbf{r_1},\mathbf{r_2})|\mathbf{r_1},t_2\right\}\label{e-value2},
\end{align}
where the expectation is taken with respect to $(y_2,\mathbf{x})$. Finally, the optimal value for the first read threshold is
\[
t_1^\star = \arg\max_{t_1}E\left\{V_1(t_1,y_1)|t_1\right\}.
\]

These policies can be computed offline and then programmed in the memory controller. Typical controllers have multiple modes tailored towards different conditions in terms of number of PE cycles, whether an upper or lower page is being read, etc. Each of these modes would have its own prior distributions for $(\mu_1,\mu_2,\sigma_1,\sigma_2)$, and would result in a different policy determining where to perform each read based on the previous results. Each policy can be stored as a partition of the feasible reads, and value functions can be discarded, so memory requirements are very reasonable. Section~\ref{s-NumResults} presents an example illustrating this scheme.

\newt{Just like in Section~\ref{ss-paramEstimation}, the number of reads can be reduced if some of the noise parameters are known or enough prior information is available. The same backward recursion could be used to optimize the choice of thresholds, but with fewer steps.}

\section{Extensions}
\label{s-Extension}

Most of the paper has assumed that cells can only store two voltage levels, with their voltages following Gaussian distributions. This framework was chosen because it is the most widely used in the literature, but the method described can easily be extended to memories with more than two levels and non-Gaussian noise distributions.

Section~\ref{s-sysmodel} explained how each wordline in a MLC (two bits per cell, four levels) or TLC (three bits per cell, eight levels) memory is usually divided into two or three pages which are read independently as if the memory was SLC. In that case, the proposed method can be applied without any modifications. However, if the controller is capable of simultaneously processing more than two levels per cell, it is possible to accelerate the noise estimation by reducing the number of reads. MLC and TLC memories generally have dedicated hardware that performs multiple reads in the ranges required to read the upper pages and returns a single binary value. For example, reading the upper page of a TLC memory with the structure illustrated in Fig.~\ref{f-grayMapping} requires four reads with thresholds (A, C, E, G) but cells between A and C would be indistinguishable from cells between E and G; all of them would be read as 0. However, one additional read of the lower page (D threshold) would allow the controller to tell them apart.

Performing four reads $(t_1,\ldots,t_4)$ on the upper page of a TLC memory would entail comparing the cell voltages against 16 different thresholds but obtaining only four bits of information for each cell. The means and variances in Eqs.~(\ref{e-estimators})-(\ref{e-estimators2-b}) would correspond to mixtures of all the levels storing the same bit value, assumed to be approximately Gaussian. The same process would then be repeated for the middle and lower page. A better approach, albeit more computationally intensive, would be to combine reads from all three pages and estimate each level independently. Performing one single read of the lower page (threshold D), two of the middle page (each involving two comparisons, with thresholds B and F) and three of the upper page (each involving four comparisons, with thresholds A, C, E, G) would theoretically provide more than enough data to estimate the noise in all eight Gaussian levels. A similar process can be used for MLC memories performing, for example, two reads of the lower page and three of the upper page.

Hence, five page reads are enough to estimate the noise mean and variance in all 4 levels of an MLC memory and 6 page reads are enough for the 8 levels in a TLC memory. Other choices for the pages to be read are also possible, but it is useful to consider that lower pages have smaller probabilities of error, so they often can be successfully decoded with fewer reads. Additional reads could provide more precise estimates and better LLR values for LDPC decoding.

\newt{There are papers suggesting that a Gaussian noise model might not be accurate for some memories~\cite{parnell2014modelling}.} The proposed scheme can also be extended to other noise distributions, as long as they can be characterized by a small number of parameters. Instead of the $Q$-function in Eq.~(\ref{e-qfunct}), the estimation should use the cumulative density function (cdf) for the corresponding noise distribution. \newt{For example, if the voltage distributions followed a Laplace instead of Gaussian distribution, Eq.~(\ref{e-mixGaussCdf}) would become
\begin{equation}
y_i=\frac{1}{2}-\frac{1}{4}e^{-\frac{t_i-\mu_1}{b_1}}+\frac{1}{2}e^{-\frac{t_i-\mu_2}{b_2}}+n_{y_i},
\end{equation}
for $\mu_1\leq t_i\leq \mu_2$ and the estimator $\hat{b_1}$ of $b_1$ would become
\begin{equation}
\widehat{b_1}=\frac{t_2-t_1}{\log(1-2y_1)-\log(1-2y_2)}
\end{equation}
when $t_1$, $t_2$ are significantly smaller than $\mu_2$. Similar formulas can be found to estimate the other parameters.}

\section{Numerical results}
\label{s-NumResults}

This section presents simulation results evaluating the performance of the dynamic programming algorithm proposed in Section~\ref{s-dynamicProgramming}. Two scenarios will be considered, corresponding to a fresh page with BER$(t^\star)=0.0015$ and a worn-out page with BER$(t^\star)=0.025$. The mean voltage values for each level will be the same in both scenarios, but the standard deviations will differ. Specifically, $\mu_1=1$ and $\mu_2=2$ for both pages, but the fresh page will be modeled using $\sigma_1=0.12$ and $\sigma_2=0.22$, while the worn page will be modeled using $\sigma_1=0.18$ and $\sigma_2=0.32$. These values, however, are unknown to the controller. The only information that it can use to choose the read locations are uniform prior distributions on $\mu_1$, $\mu_2$, $\sigma_1$, and $\sigma_2$, identical for both the fresh and the worn-out pages. Specifically, $\mu_1$ is known to be in the interval $(0.75,1.25)$, $\mu_2$ in $(1.8,2.1)$, $\sigma_1$ in $(0.1,0.24)$ and $\sigma_2$ in $(0.2,0.36)$.

For each scenario, three different strategies for selecting the read thresholds were evaluated. The first strategy, $S_1$, tries to obtain accurate noise estimates by spreading out the reads. The second strategy, $S_2$,  concentrates all of them on the uncertainty region, attempting to attain highly informative LLR values. Finally, the third strategy, $S_3$, follows the optimal policy obtained by the dynamic programming recursion proposed in Section~\ref{s-dynamicProgramming}, with $C_{P,\hat{P}}$ as reward function. The three strategies are illustrated in Fig.~\ref{f-strategies} and the results are summarized in Table~\ref{t-results}, but before proceeding to their analysis we describe the process employed to obtain $S_3$.

The dynamic programming scheme assumed that read thresholds were restricted to move in steps of $0.04$, and quantized all cdf measurements also in steps of $0.04$ (making the noise $n_{y}$ from Eq.~(\ref{e-mixGaussCdf}) uniform between $-0.02$ and $0.02$). Starting from these assumptions, Eqs.~(\ref{e-policy4}) and (\ref{e-value4}) were used to find the optimal policy $\pi_4$ and expected value $V_3$ for all meaningful combinations of $(t_1,y_1,t_2,y_2,t_3,y_3)$, which were in the order of $10^6$ (very reasonable for offline computations). The value function $V_3$ was then used in the backward recursion to find the policies and values for the first three reads as explained in Section~\ref{s-dynamicProgramming}. The optimal location for the first read, in terms of maximum expected value for $I(X;Y)-D(P\|\hat{P})$ after all four reads, was found to be $t_1^\star=1.07$. This read resulted in $y_1=0.36$ for the fresh page and $y_1=0.33$ for the worn page. The policy $\pi_2$ dictated that $t_2=0.83$ for $y_1\in(0.34,0.38)$, and $t_2=1.63$ for $y_1\in(0.3,0.34)$, so those were the next reads in each case. The third and fourth read thresholds $t_3$ and $t_4$ were chosen similarly according to the corresponding policies.

Finally, as depicted in Fig.~\ref{f-strategies}, the read thresholds were
\begin{itemize}
\item $S_1$: $\mathbf{t}=(0.85,\ 1.15,\ 1.75,\ 2.125)$.
\item $S_2$: $\mathbf{t}=(1.2,\ 1.35,\ 1.45,\ 1.6)$.
\item $S_3$ (fresh page): $\mathbf{t}=(1.07,\ 0.83,\ 1.79,\ 1.31)$ resulting in $\mathbf{y}=(0.36,\ 0.04,\ 0.58,\ 0.496)$, respectively.
\item $S_3$ (worn page): $\mathbf{t}=(1.07,\ 1.63,\ 1.19,\ 1.43)$ resulting in $\mathbf{y}=(0.33,\ 0.56,\ 0.43,\ 0.51)$, respectively.
\end{itemize}
For the fresh page, the policy dictates that the first three reads should be performed well outside of the uncertainty region, so as to obtain good estimates of the means and variances. Then, the fourth read is performed as close as possible to the BER-minimizing threshold. Since the overlap between both levels is very small, soft decoding would barely provide any gain over hard decoding. Picking the first three reads for noise characterization regardless of their value towards building LLRs seems indeed to be the best strategy. For the worn-out page, the policy attempts to achieve a trade-off by combining two reads away from the uncertainty region, good for parameter estimation, with another two inside it to improve the quality of the LLR values used for soft decoding.
\begin{figure}[htb]
  \centering
  \begin{tabular}{cc}
  \epsfig{figure=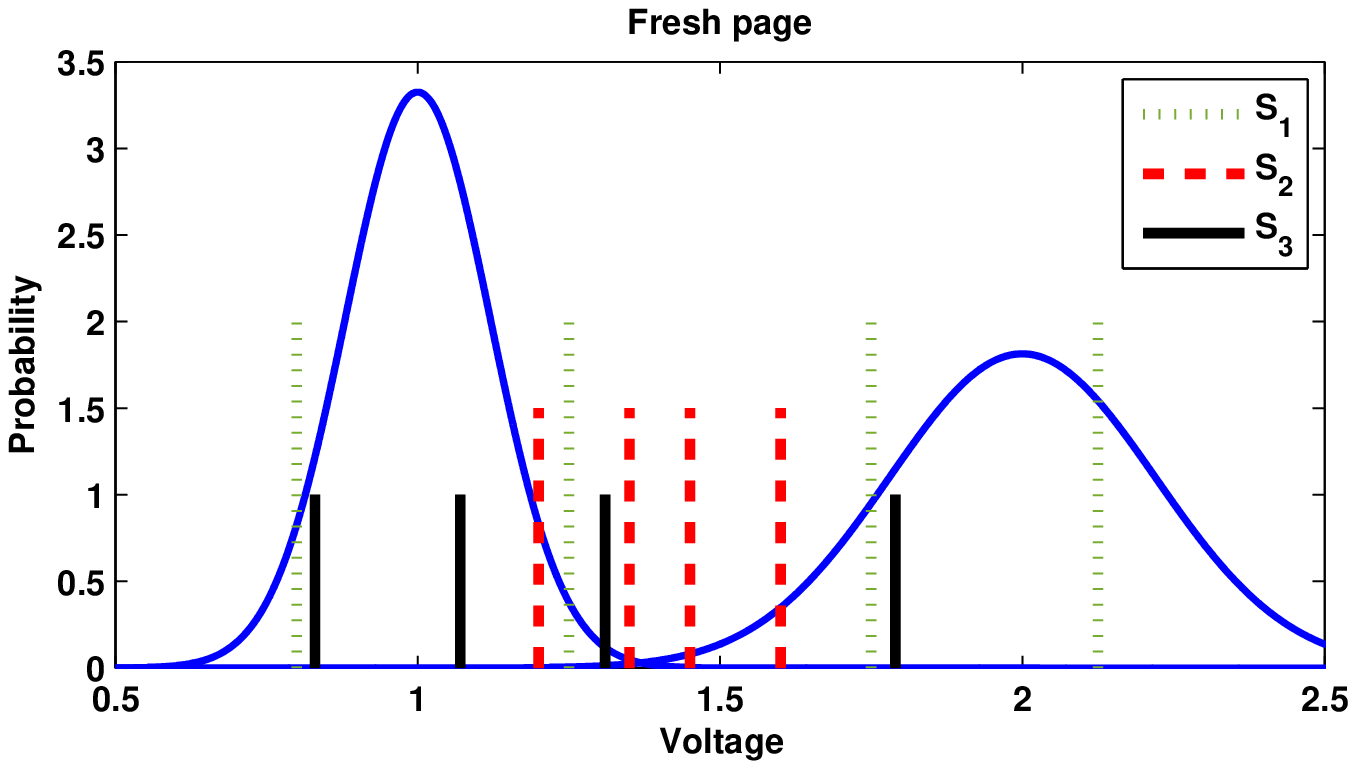,width=8cm} & \epsfig{figure=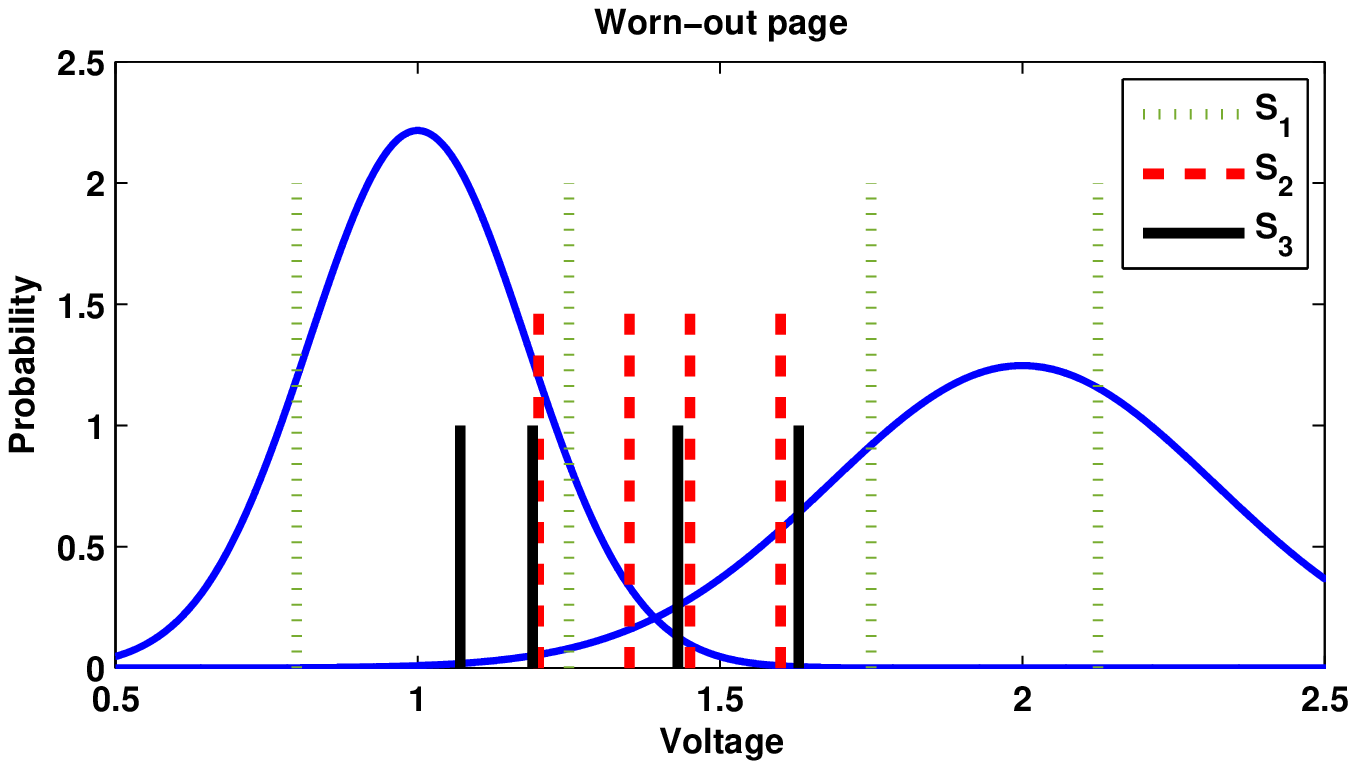,width=8cm}
  \end{tabular}
\caption{Read thresholds for strategies $S_1$, $S_2$ and $S_3$ for a fresh and a worn-out page.}
\label{f-strategies}
\end{figure}

\begin{table}
\centering
\begin{tabular}{c c}
\begin{tabular}{|c|c c c|}
\hline
FRESH PAGE & $S_1$ & $S_2$ & $S_3$\\
\hline
$|\hat\mu-\mu|/\mu$ & 0.004 & 0.182 & 0.012\\
$|\hat\sigma-\sigma|/\sigma$ & 0.03 & 0.91 & 0.12 \\
$|\widehat{t^\star}-t^\star|/t^\star$ & 0.01 & 0.07 & 0.02\\
$|\mathrm{BER}(\widehat{t^\star})-\mathrm{BER}(t^\star)|/\mathrm{BER}(t^\star)$ & 0.1 & 1.4 & 0.11\\
LDPC fail rate & 1 & 0.15 & 0\\
Genie LDPC fail rate & 1 & 0 & 0\\
\hline
\end{tabular}
&
\begin{tabular}{|c|c c c|}
\hline
OLD PAGE & $S_1$ & $S_2$ & $S_3$\\
\hline
$|\hat\mu-\mu|/\mu$ & 0.005 & 0.053 & 0.021\\
$|\hat\sigma-\sigma|/\sigma$ & 0.03 & 0.27 & 0.13 \\
$|\widehat{t^\star}-t^\star|/t^\star$ & 0.006 & 0.015 & 0.011\\
$|\mathrm{BER}(\widehat{t^\star})-\mathrm{BER}(t^\star)|/\mathrm{BER}(t^\star)$ & 0.003 & 0.009 & 0.007 \\
LDPC fail rate & 1 & 0.19 & 0.05\\
Genie LDPC fail rate& 1 & 0 & 0.01\\
\hline
\end{tabular}
\end{tabular}
\caption{Trade-off between BER and LDPC failure rate.}\label{t-results}\vspace{-2cm}
\end{table}

Table~\ref{t-results} shows the relative error in our estimates and sector failure rates averaged over 5000 simulation instances, with read noise $n_{y_i}$, $i=1,\ldots,4$ uniformly distributed between $-0.02$ and $0.02$. The first three rows show the relative estimation error of the mean, variance, and optimal threshold. It can be observed that $S_1$ provides the lowest estimation error, while $S_2$ produces clearly wrong estimates. The estimates provided by $S_3$ are noisier than those provided by $S_1$, but are still acceptable. The relative increase in BER when reading at $\widehat{t^\star}$ instead of at $t^\star$ is shown in the fourth row of each table. It is worth noting that the BER$(\widehat{t^\star})$ does not increase significantly, even with inaccurate mean and variance estimates. This validates the derivation in Section~\ref{ss-errorprop}.

Finally, the last two rows on each table show the failure rate after 20 iterations of a min-sum LDPC decoder for two different methods of obtaining soft information. The LDPC code had $18\%$ redundancy and codeword length equal to 35072 bits. The fifth row corresponds to LLR values obtained using the mean and variance estimates from the Progressive Read Algorithm and the last row, labeled ``Genie LDPC", corresponds to using the actual values instead of the estimated ones. It can be observed that strategy $S_1$, which provided very accurate estimates, always fails in the LDPC decoding. This is due to the wide range of cell voltages that fall between the middle two reads, being assigned an LLR value close to $0$. The fact that the ``Genie LDPC"  performs better with $S_2$ than with $S_3$ shows that the read locations chosen by the former are better. However, $S_3$ provides lower failure rates in the more realistic case where the means and variances need to be estimated using the same reads used to produce the soft information.

In summary, $S_3$ was found to be best from an LDPC code point of view and $S_1$ from a pure BER-minimizing perspective. $S_2$ as proposed in~\cite{wang2011soft} is worse in both cases unless the voltage distributions are known. \newt{When more than four reads are allowed, all three schemes perform similarly. After the first four reads, all the strategies have relatively good estimates for the optimal threshold. Subsequent reads are located close to the optimal threshold, achieving small BER. Decoding failure rates are then limited by the channel capacity, rather than by the location of the reads.} 

\section{Conclusion}
\label{s-conclusion}
NAND flash controllers often require several re-reads using different read thresholds to recover host data in the presence of noise. In most cases, the controller tries to guess the noise distribution based on the number of PE cycles and picks the read thresholds based on that guess. However, unexpected events such as excessive leakage or charge trapping can make those thresholds suboptimal. This paper proposed algorithms to reduce the total read time and sector failure rate by using a limited number of re-reads to estimate the noise and improve the read thresholds.

The overall scheme will work as follows. First, the controller will generally have a prior estimation of what a good read threshold might be. It will read at that threshold and attempt a hard-decoding of the information. If the noise is weak and the initial threshold was well chosen, this decoding will succeed and no further processing will be needed. In cases when this first decoding fails, the controller will perform additional reads to estimate the mean and/or variance of the voltage values for each level. These estimates will in turn be used to estimate the minimum achievable BER and the corresponding optimal read threshold. The flash controller then decides whether to perform an additional read with this threshold to attempt hard decoding again, or directly attempt a more robust decoding of the information, for example LDPC, leveraging the reads already performed to generate the soft information.

The paper proposes using a dynamic programming backward recursion to find a policy for progressively picking the read thresholds based on the prior information available and the results from previous reads. This scheme will allow us to find the thresholds that optimize an arbitrary objective. Controllers using hard decoding only (e.g., BCH) may wish to find the read threshold providing minimum BER, while those employing soft decoding (e.g., LDPC) will prefer to maximize the capacity of the resulting channel. The paper provides an approximation for the (symmetric and mismatched) capacity of the channel and presents simulations to illustrate the performance of the proposed scheme in such scenarios.

\bibliographystyle{IEEEtran}
\bibliography{AdaptiveReadThresholds}

\section{Appendix}
\begin{IEEEproof}
\textit{(Theorem 1)}
The proof is very similar to that for Shannon's Channel Coding Theorem, but a few changes will be introduced to account for the mismatched decoder. Let $X\in\{1,2\}^n$ denote the channel input and $Y\in\mathcal{Y}^n$ the channel output, with $X_i$ and $Y_i$ denoting their respective components for $i=1,\ldots,n$. Throughout the proof, $\hat{P}(A)$ will denote the estimate for the probability of an event $A$ obtained using the transition probabilities $\widehat{p_{ij}}$, $i=1,2$, $j=1,\ldots,|\mathcal{Y}|$, to differentiate it from the exact probability $P(A)$ obtained using transition probabilities $p_{ij}$, $i=1,2$, $j=1,\ldots,|\mathcal{Y}|$. The inputs are assumed to be symmetric, so $\hat{P}(X)=P(X)$ and $\hat{P}(X,Y)=\hat{P}(Y|X)P(X)$.

 We start by generating 
 $2^{nR}$ random binary sequences of length $n$ to form a random code $\mathcal{C}$ with rate $R$ and length $n$. 
 After revealing the code $\mathcal{C}$ to both the sender and the receiver, a codeword $\mathbf{x}$ is chosen at random among those in $\mathcal{C}$ and transmitted. The conditional probability of receiving a sequence $\mathbf{y}\in\mathcal{Y}^n$ given the transmitted codeword $\mathbf{x}$ 
 is given by
$P(Y=\mathbf{y}|X=\mathbf{x})=\prod_{i=1}^n p_{x_iy_i}$,
where $x_i$ and $y_i$ denote the $i$-th components of $\mathbf{x}$ and $\mathbf{y}$, respectively.

The receiver then attempts to recover the codeword $\mathbf{x}$ that was sent. However, the decoder does not have access to the exact transition probabilities $p_{ij}$ and must use the estimated probabilities $\widehat{p_{ij}}$ instead. When $p_{ij}=\widehat{p_{ij}}\quad\forall i,j$, the optimal decoding procedure is maximum likelihood decoding (equivalent to maximum {\it a posteriori} decoding, since inputs are equiprobable). In maximum likelihood decoding, the decoder forms the estimate
 $\mathbf{\hat{x}}=\arg\max_{\mathbf{x}\in\mathcal{C}}\hat{P}(\mathbf{y}|\mathbf{x})$,
 where $\hat{P}(Y=\mathbf{y}|X=\mathbf{x})=\prod_{i=1}^n \widehat{p_{x_iy_i}}$ is the estimated likelihood of $\mathbf{x}$, given $\mathbf{y}$ was received.

 Denote by $\hat{A}_\epsilon^{(n)}$ the set of length-$n$ sequences $\{(\mathbf{x},\mathbf{y})\}$ whose estimated empirical entropies are $\epsilon$-close to the typical estimated entropies:
 \begin{eqnarray}\label{e-typical_hat}
\hat{A}_\epsilon^{(n)} & = & \left\{(\mathbf{x},\mathbf{y})\in\{1,2\}^n\times\mathcal{Y}^n:\right.\\
 & & \left|-\frac{1}{n}\log P(X=\mathbf{x})-1\right|<\epsilon,\\
& &\left|-\frac{1}{n}\log \hat{P}(Y=\mathbf{y})-\mu_Y\right|<\epsilon,\\
& &\left.\left|-\frac{1}{n}\log \hat{P}(X=\mathbf{x},Y=\mathbf{y})-\mu_{XY}\right|<\epsilon\right\},
\end{eqnarray}
where $\mu_Y$ and $\mu_{XY}$ represent the expected values of $-\frac{1}{n}\log \hat{P}(Y)$ and $-\frac{1}{n}\log \hat{P}(X,Y)$, respectively, and the logarithms are in base 2. Hence,
\begin{align}
\mu_Y 
& = -\frac{1}{n}\sum_{i=1}^n\sum_{k=1}^{|\mathcal{Y}|}P(Y_i=k)\log \hat{P}(Y_i=k)\\
& = -\sum_{k=1}^{|\mathcal{Y}|}\frac{p_{1k}+p_{2k}}{2}\log\left(\frac{\widehat{p_{1k}}+\widehat{p_{2k}}}{2}\right),\\
\mu_{XY} 
& = -\frac{1}{n}\sum_{i=1}^n\sum_{b=1}^2\sum_{k=1}^{|\mathcal{Y}|}P(X_i=b,Y_i=k)\log \hat{P}(X_i=b,Y_i=k)\\
& = -\sum_{k=1}^{|\mathcal{Y}|}\left(\frac{p_{1k}}{2}\log\left(\frac{\widehat{p_{1k}}}{2}\right)+\frac{p_{2k}}{2}\log\left(\frac{\widehat{p_{2k}}}{2}\right)\right),
\end{align}
where the exact transition probabilities are used as weights in the expectation and the estimated ones are the variable values. Particularly, $(\mathbf{x},\mathbf{y})\in\hat{A}_\epsilon^{(n)}$ implies that $\hat{P}(Y=\mathbf{y}|X=\mathbf{x})>2^{n(1-\mu_{XY}-\epsilon)}$ and $\hat{P}(Y=\mathbf{y})<2^{-n(\mu_{Y}-\epsilon)}$. We will say that a sequence $\mathbf{x}\in\{1,2\}^n$ is in  $\hat{A}_\epsilon^{(n)}$ if it can be extended to a sequence $(\mathbf{x},\mathbf{y})\in \hat{A}_\epsilon^{(n)}$, and similarly for $\mathbf{y}\in\mathcal{Y}^n$.

First we show that with high probability, the transmitted and received sequences $(\mathbf{x},\mathbf{y})$ are in the $\hat{A}_\epsilon^{(n)}$ set. The weak law of large numbers states that for any given $\epsilon>0$, there exists $n_0$, such that for any codeword length $n>n_0$
\begin{align}
P\left(\left|-\frac{1}{n}\log P(X=\mathbf{x})-1\right|\geq\epsilon\right) & < \frac{\epsilon}{3},\\
P\left(\left|-\frac{1}{n}\log \hat{P}(Y=\mathbf{y})-\mu_Y\right|\geq\epsilon\right)& < \frac{\epsilon}{3},\\
P\left(\left|-\frac{1}{n}\log \hat{P}(X=\mathbf{x},Y=\mathbf{y})-\mu_{XY}\right|\geq\epsilon\right) & < \frac{\epsilon}{3}.
\end{align}
Applying the union bound to these events shows that for $n$ large enough, $P\left((\mathbf{x},\mathbf{y})\notin\hat{A}_\epsilon^{(n)}\right)<\epsilon$.

 When a codeword $\mathbf{x}\in\{1,2\}^n$ is transmitted and $\mathbf{y}\in\mathcal{Y}^n$ is received, an error will occur if there exists another codeword $\mathbf{z}\in\mathcal{C}$ such that $\hat{P}(Y=\mathbf{y}|X=\mathbf{z})\geq\hat{P}(Y=\mathbf{y}|X=\mathbf{x})$. The estimated likelihood of $\mathbf{x}$ is greater than $2^{n(1-\mu_{XY}-\epsilon)}$ with probability at least $1-\epsilon$, as was just shown. The other $nR-1$ codewords in $\mathcal{C}$ are independent from the received sequence. For a given $\mathbf{y}\in\hat{A}_\epsilon^{(n)}$, let $S_\mathbf{y}=\left\{\mathbf{x}\in\{1,2\}^n:\hat{P}(Y=\mathbf{y}|X=\mathbf{x})\geq 2^{n(1-\mu_{XY}-\epsilon)}\right\}$ denote the set of input sequences whose estimated likelihood is greater than $2^{n(1-\mu_{XY}-\epsilon)}$. Then
 \begin{align}
 1 & = \sum_{\mathbf{x}\in\{1,2\}^n}\hat{P}(X=\mathbf{x}|Y=\mathbf{y})\\
 & > \sum_{\mathbf{x}\in S_{\mathbf{y}}}\hat{P}(Y=\mathbf{y}|X=\mathbf{x})\frac{P(X=\mathbf{x})}{\hat{P}(Y=\mathbf{y})}\\
 & > |S_\mathbf{y}|2^{n(1-\mu_{XY}-\epsilon)}2^{-n}2^{n(\mu_Y-\epsilon)}
 \end{align}
which implies $|S_\mathbf{y}|<2^{n(\mu_{XY}-\mu_Y+2\epsilon)}$ for all $\mathbf{y}\in\hat{A}_\epsilon^{(n)}$.

If $(\mathbf{x},\mathbf{y})\in\hat{A}_\epsilon^{(n)}$, any other codeword causing an error must be in $S_\mathbf{y}$. Let $E_i$, $i=1,\ldots,nR-1$ denote the event that the $i$-th codeword in the codebook $\mathcal{C}$ is in $S_\mathbf{y}$, and $F$ the event that $(\mathbf{x},\mathbf{y})$ are in $\hat{A}_\epsilon^{(n)}$. The probability of error can be upper bounded by
\begin{align}
P(\mathbf{\hat{x}}\neq\mathbf{x}) & = P(F^c)P(\mathbf{\hat{x}}\neq\mathbf{x}|F^c) + P(F)P(\mathbf{\hat{x}}\neq\mathbf{x}|F)\\
& \leq \epsilon P(\mathbf{\hat{x}}\neq\mathbf{x}|F^c) + \sum_{i=1}^{2^{nR}-1}P(E_i|F)\\
& \leq \epsilon + 2^{nR}|S_\mathbf{y}|2^{-n}\\
& \leq \epsilon + 2^{n(R+\mu_{XY}-\mu_Y-1+2\epsilon)}
\end{align}

Consequently, as long as
\begin{equation}
R < \frac{1}{2}\sum_{k=1}^{|\mathcal{Y}|}\left(p_{1k}\log\left(\widehat{p_{1k}}\right)+p_{2k}\log\left(\widehat{p_{2k}}\right)\right)-(p_{1k}+p_{2k})\log\left(\frac{\widehat{p_{1k}}+\widehat{p_{2k}}}{2}\right),
\end{equation}
for any $\delta>0$, we can choose $\epsilon$ and $n_\epsilon$ so that for any $n>n_\epsilon$ the probability of error, averaged over all codewords and over all random codes of length $n$, is below $\delta$. By choosing a code with average probability of error below $\delta$ and discarding the worst half of its codewords, we can construct a code of rate $R-\frac{1}{n}$ and maximal probability of error below $2\delta$, proving the achievability of any rate below the bound $C_{P,\hat{P}}$ defined in Eq.~(\ref{e-capacityHat}). This concludes the proof.
\end{IEEEproof}
\end{document}